\documentclass[11pt]{article}
\usepackage{amsmath, amssymb, amsthm, fullpage, natbib}
\usepackage{enumitem}

\title{Irreversible Failure Reverses the Value of Information
}
\author{
Nicholas H. Kirk$^{1,2}$\thanks{Email: \texttt{sbsdfs26009@said.oxford.edu}}\\
$^{1}$ Markovian Group\\
$^{2}$ Sa\"id Business School, University of Oxford
}

\date{}

\newtheorem{definition}{Definition}
\newtheorem{assumption}{Assumption}
\newtheorem{proposition}{Proposition}
\newtheorem{theorem}{Theorem}
\newtheorem{lemma}{Lemma}
\newtheorem{remark}{Remark}

\newcommand{\E}{\mathbb{E}}

\newcommand{\RR}{\mathbb{R}}

\begin{document}
\maketitle

\begin{abstract}
We study dynamic games with hidden states and absorbing failure, where belief-driven actions can trigger irreversible collapse. In such environments, equilibria that sustain activity generically operate at the boundary of viability. We show that this geometry endogenously reverses the value of information: greater informational precision increases the probability of collapse on every finite horizon. We formalize this mechanism through a limit-viability criterion, and model opacity as a strategic choice of the information structure via Blackwell garbling. When failure is absorbing, survival values become locally concave in beliefs, implying that transparency destroys equilibrium viability while sufficient opacity restores it. In an extended game where agents choose the information structure ex ante, strictly positive opacity is necessary for equilibrium survival.

The results identify irreversible failure—not coordination, misspecification, or ambiguity—as a primitive force generating an endogenous demand for opacity in dynamic games.
\end{abstract}



\section{Introduction}

Information is widely viewed as stabilizing economic behavior: greater transparency improves
learning, disciplines actions, and enhances efficiency. This presumption fails when failure is
irreversible. In many dynamic environments—such as bankruptcy, sovereign default, regime collapse,
or revolt—even small belief fluctuations can trigger actions that permanently destroy the system.
In such settings, greater informational precision can increase, rather than reduce, the likelihood
of catastrophic outcomes.

This paper provides a structural explanation for this reversal. We study dynamic games with hidden
states and absorbing failure, in which actions taken under sufficiently adverse beliefs can trigger
permanent collapse. Absorbing failure fundamentally alters the geometry of continuation values:
it transforms otherwise smooth payoff functions into cliff-shaped survival values. 
Prior work on sustainable behavior under irreversible failure notes that equilibria sustaining activity typically operate at the boundary of viability, leaving no interior buffer against
collapse \citep{kirk2025sustainable}. 
This mechanism has a direct implication for behavior. 
When information is weaponized in this sense, greater informational precision strictly lowers
expected survival near the viability boundary. As a result, agents optimally restrict the
informativeness of the observation channel, not because information is distorted, but because
opacity concavifies survival values and suppresses belief realizations that trigger irreversible
failure. By endogenously garbling information, agents reduce the dispersion of posterior beliefs and suppress
the likelihood of realizations that induce collapse. Strategic opacity is therefore not modeled as misreporting or false statements; it is a strategic
choice of the information structure (garbling) that can be deceptive in effect while remaining
fully Bayesian.

The paper formalizes this logic by modeling opacity as a strategic choice of the information
structure via Blackwell garbling. Outcomes are evaluated using a limit-viability criterion:
equilibria must avoid collapse with arbitrarily high probability on every finite horizon as
informational noise vanishes. This criterion is appropriate for environments with absorbing failure,
where almost-sure survival is unattainable under full-support noise and expected-payoff criteria
can mask fragility. 

Our main results establish that in weaponized-information environments, increased opacity weakly
expands the set of limit-viable outcomes. We show that transparency can induce global collapse:
there exist environments in which no limit-viable equilibrium exists under maximal transparency,
while sufficiently high opacity restores equilibrium existence. When agents can choose opacity
ex ante, any equilibrium that satisfies limit-viability must feature strictly positive opacity. In standard Bayesian persuasion, the curvature of the value function in beliefs is taken as primitive, and information is harmful only when payoffs are exogenously concave.

This paper identifies a broad and economically central
class of dynamic environments in which concavity is endogenously generated by irreversible failure.
Once this geometry is recognized, persuasion logic follows as a consequence rather than an
assumption: opacity emerges as the optimal concavification of survival values.

The core mechanism can be summarized succinctly. In dynamic environments with absorbing failure, the mapping from beliefs to continuation values is discontinuous. When equilibrium actions switch discretely near a failure threshold, posterior concentration under precise information increases exposure to irreversible collapse. Opacity mitigates this risk by suppressing belief realizations that cross the survival cliff. The contribution is not a new equilibrium concept, but a structural impossibility result linking transparency and viability in dynamic games with irreversible failure.

The paper proceeds as follows. Section~\ref{sec:related} situates the contribution relative to the
literature on information design, global games, and fragility. Section~\ref{sec:environment}
describes the environment and the modeling of opacity. Section~\ref{sec:viability} introduces the
limit-viability criterion. Section~\ref{sec:weaponized} formalizes weaponized information and Section~\ref{sec:results} derives the main results. Sections~\ref{sec:benchmark} and~\ref{sec:colonial} present illustrative
applications. Section~\ref{sec:conclusions} concludes.

\section{Related Work}
\label{sec:related}
\paragraph{Relation to global games and crisis fragility.}
A large literature shows that small noise can generate crisis dynamics in coordination environments,
either by selecting a unique equilibrium \citep{morris1998coordination} or, when information is endogenously
generated (e.g., through prices or policy), by sustaining multiplicity and volatility
\citep{AngeletosHellwigPavan2006,AngeletosWerning2006}. The mechanism in the present paper is distinct.
Unlike models in which noise induces panic through coordination effects, the present paper
identifies a class of environments in which \emph{greater precision} itself is destabilizing,
even absent multiplicity or strategic uncertainty.
We do not study equilibrium selection or multiplicity driven by strategic complementarities
under exogenous noise.
Instead, absorbing failure endogenously induces local concavity of survival values in beliefs,
so that greater informational precision can increase exposure to collapse even holding the
strategic environment fixed.
Moreover, the key behavioral margin here is endogenous control of the information structure
(opacity via Blackwell garbling), evaluated under a limit-viability criterion, rather than the
comparative statics of equilibrium thresholds under a fixed information structure.

\paragraph{Global games and vanishing noise.}
The paper is related to the global games literature, where vanishing private noise refines equilibrium
selection by breaking common knowledge and eliminating multiplicity \citep{carlsson1998coordination,morris1998coordination}.
In that literature, increased signal precision is stabilizing: as noise vanishes, behavior converges
smoothly to a selected equilibrium. The mechanism here is different. In global games, noise selects among equilibria; here, precision eliminates equilibrium existence.
With absorbing failure, the mapping from beliefs to outcomes is discontinuous. As informational noise vanishes, posterior concentration can
instead activate failure-triggering actions, causing collapse with positive probability on finite
horizons. Rather than selecting equilibria, transparency can destroy viability altogether.

\paragraph{Disclosure, signaling, and information design.}
Unlike signaling or disclosure models, opacity here does not involve misreporting, persuasion, or
strategic distortion of beliefs \citep{grossman1980disclosure,milgrom1981good}.
Beliefs remain Bayesian and correctly specified. Opacity is modeled as an endogenous choice of the
information structure itself—via Blackwell garbling of signals—prior to play \citep{blackwell1953comparison,kamenica2011persuasion}.
The mechanism is not reputational manipulation or persuasion, but the avoidance of belief-driven
feasibility violations in fragile systems with irreversible failure.

\paragraph{Robust control and ambiguity aversion.}
The results are distinct from robust control and ambiguity-based approaches, which introduce
conservatism through preferences, worst-case beliefs, or model misspecification \citep{hansen2008robust,epstein2003ambiguity}.
Here, agents are fully Bayesian and correctly specified. Opacity is chosen not because agents dislike
uncertainty, but because greater transparency increases exposure to irreversible failure in
environments where information is weaponized.

\paragraph{Viability, sustainability, and absorbing states.}
The focus on absorbing failure connects to work on viability and sustainability in dynamic systems
\citep{aubin1991viability,altman1999constrained}. Absorbing states and irreversible transitions are also
central in models of default, regime collapse, and institutional breakdown \citep{acemoglu2005institutions}.
The contribution of this paper is to show that, under irreversible collapse, the information structure
itself becomes a strategic margin: agents may optimally coarsen information to preserve survival as
informational frictions vanish.

\paragraph{Relation to learning and misspecification.}
A growing literature explains fragility through learning dynamics, model misspecification,
or robustness concerns, in which agents act conservatively because beliefs are distorted or
because worst-case reasoning is introduced \citep{epstein2003ambiguity,hansen2008robust}.
Related work on rare disasters and tail risks emphasizes belief updating under infrequent
but catastrophic events, showing how pessimistic learning can generate persistent
macroeconomic effects \citep{KozlowskiVeldkampVenkateswaran2020}.
More recently, misspecification-based equilibrium concepts formalize how agents behave when
their subjective models are systematically incorrect \citep{EspondaPouzo2016}.

The mechanism in the present paper is distinct.
Agents are fully Bayesian, correctly specified, and share common knowledge of the environment.
No ambiguity preferences \citep{epstein2003ambiguity}, worst-case distortions \citep{hansen2008robust},
or misspecified learning dynamics \citep{EspondaPouzo2016} are assumed.
The reversal in the value of information arises even when beliefs converge under vanishing noise:
greater informational precision increases exposure to irreversible failure by activating
belief-triggered actions near viability boundaries.

Accordingly, strategic opacity is not a response to epistemic uncertainty or pessimistic learning
\citep{KozlowskiVeldkampVenkateswaran2020}, but a rational adjustment of the information structure
in environments where correct beliefs themselves can be destructive.
This isolates irreversibility—rather than learning or misspecification—as the primitive force
driving the endogenous demand for opacity.

\paragraph{Relation to work on viability selection.}
Viability selection asks which equilibria survive when perturbations of various nature are exogenous; strategic opacity asks how agents optimally adjust the
observation channel itself once such fragility is recognized. The two are complementary: selection
filters behavior given an information environment, while opacity changes the information environment.

\paragraph{Behavioral interpretation.}
Although the mechanism concerns information, the contribution is behavioral rather than
epistemic. Agents do not misperceive, distort, or manipulate beliefs. Instead, they choose
how much information to process or reveal as part of their strategic behavior. In environments
with irreversible failure, limiting information exposure becomes a rational form of self-protection.

The paper is also related to earlier work on sustainability and equilibrium feasibility in dynamic
games with irreversible failure. Studies on Sustainable Exploitation Equilibrium (SEE) in \citep{kirk2025sustainable} characterize equilibrium behavior that maintains viability under exact or stable beliefs.
The present paper asks a distinct question. Taking the boundary nature of sustainable equilibria as given, the paper studies how agents
optimally shape the information structure itself when viability is fragile with respect to informational precision. In this sense, SEE analyzes feasibility conditional on beliefs,
while strategic opacity analyzes behavior over the information structure when feasibility is
endogenously threatened by belief dispersion.

Accordingly, the paper complements the crisis/global-games literature by identifying an information-design
motive for \emph{less} precision that does not rely on multiplicity, but on irreversible failure. Unlike global games, persuasion, or robust-control models, the present paper shows that
even fully Bayesian agents with correct beliefs may require less information for equilibrium
existence once failure is irreversible.

Table~\ref{tab:comparison-mechanisms} summarizes how this mechanism differs sharply from both global games and
Bayesian persuasion: information here does not select equilibria or reallocate behavior,
but determines whether equilibrium survival is feasible at all.

\section{Environment}
\label{sec:environment}

Time is discrete, indexed by $t = 0,1,2,\ldots$. There is a finite set of players
$N = \{1,\dots,n\}$. The payoff-relevant state $s_t \in S$ evolves according to a
controlled Markov kernel $Q(\cdot \mid s_t, a_t)$ on a compact metric space $S$,
where $a_t = (a_{1t},\dots,a_{nt})$ is the action profile at date $t$.

A subset $F \subset S$ is absorbing: once $s_t \in F$, collapse occurs permanently
and all continuation payoffs are zero.

Each player $i$ has a compact action space $A_i$, and chooses actions
$a_{it} \in A_i$ each period until failure occurs. Flow payoffs are given by
measurable functions
\[
u_i : S \times A \to \mathbb{R},
\quad A := \prod_{i \in N} A_i,
\]
with the normalization that $u_i(s,a)=0$ for all $s \in F$.

Players do not observe the state $s_t$ directly. Instead, conditional on $s_t$,
they observe signals generated by an observation channel
$\kappa_{\varepsilon}(\cdot \mid s_t)$, where $\varepsilon > 0$ parametrizes
informational noise and $\varepsilon \downarrow 0$. Signals satisfy a monotone
likelihood ratio property.

Strategies are measurable mappings from private histories to actions. An
equilibrium is a Perfect Bayesian Equilibrium (PBE)\footnote{The choice of PBE is without loss for the results below; all arguments rely only on sequential rationality and Bayesian updating on the equilibrium path.} of the induced dynamic game,
with beliefs updated by Bayes’ rule wherever possible.

Let $\Delta(S)$ denote the set of Borel probability measures on $S$. Let
$b_t \in \Delta(S)$ denote the \emph{common posterior belief} about $s_t$ induced by
Bayesian updating under the equilibrium strategy profile and the public
information structure. Common posterior means that, conditional on the public
history and the observation channel, players’ beliefs about $s_t$ coincide almost
surely.

For any strategy profile $\sigma$ and any finite horizon $T < \infty$, define the
finite-horizon survival value
\[
V_T(b;\sigma) \equiv 1 - \Pr_{\sigma}(\tau_F \le T \mid b_0 = b),
\]
that is, the probability of avoiding the absorbing failure set up to time $T$
starting from belief $b$.

All results below are invariant to the number of players and do not rely on
symmetry or on specific payoff functional forms.

\subsection*{Equilibrium and Information Structure}

We impose the following regularity conditions:

\begin{assumption}[Regularity]
\label{ass:regularity}
\begin{enumerate}[label=(\roman*)]
    \item The Markov kernel $Q(\cdot \mid s, a)$ is continuous in $(s,a)$ and satisfies the Feller property.
    \item Payoff functions $u_i : S \times A \to \mathbb{R}$ are continuous and bounded.
    \item The observation channel $\kappa^\varepsilon(\cdot \mid s)$ has a density with respect to a dominating measure, satisfies the monotone likelihood ratio property globally, and converges weakly to a degenerate distribution at $s$ as $\varepsilon \downarrow 0$.
    \item The failure set $F \subset S$ is closed.
\end{enumerate}
\end{assumption}

\textbf{Common Posterior.} Under Perfect Bayesian Equilibrium with public observation channel $\kappa^\varepsilon$, all players observe signals drawn from the same conditional distribution (subject to independent idiosyncratic noise with vanishing variance). Consequently, posterior beliefs $b_t$ conditional on the public history converge in probability to a common limit as private noise vanishes. We refer to this limit as the \emph{common posterior}. All statements about beliefs and survival values are understood to hold along this common posterior path.

\textbf{Finite Horizon Payoffs.} For any finite horizon $T < \infty$ and initial belief $b \in \Delta(S)$, player $i$'s expected discounted payoff under strategy profile $\sigma$ is
\begin{equation}
    U_i^T(b; \sigma) = \mathbb{E}^\sigma \left[ \sum_{t=0}^T \delta^t u_i(s_t, a_t) \mathbf{1}_{\{t < \tau_F\}} \,\bigg|\, b_0 = b \right],
\end{equation}
where $\delta \in (0,1)$ is the common discount factor.

\section{Endogenous Opacity}

Agents choose the informativeness of the observation channel prior to play.

\begin{assumption}[Opacity as Blackwell garbling]
For $\lambda' \ge \lambda$, the channel $\kappa^{\varepsilon,\lambda'}$ is a Blackwell garbling of
$\kappa^{\varepsilon,\lambda}$. More opaque channels are weakly less informative in Blackwell order.
\end{assumption}

Higher $\lambda$ corresponds to greater opacity.

In the language of information design \citep{blackwell1953comparison,kamenica2011persuasion}, choosing $\lambda$ amounts
to choosing a feasible distribution of posteriors (with fixed mean equal to the prior). More opaque channels induce
posterior distributions that are less dispersed in convex order, which is the comparative statics relevant when survival
values are concave in beliefs (Proposition \ref{prop:jensen}).

\begin{remark}[Opacity versus falsification]
The paper models opacity as a choice of informativeness (Blackwell garbling). This is distinct from
strategic misreporting or lying, though in applications institutional control over measurement and
disclosure may serve similar objectives.
\end{remark}

\paragraph{Control of the information structure.}
Throughout the paper, opacity is a choice over the publicly available observation channel.
Individual agents do not have access to private, costless deviations in signal precision.
This corresponds to environments in which information acquisition is centralized,
technologically constrained, or institutionally governed (e.g.\ accounting standards,
statistical agencies, colonial administrations, or disclosure regimes).
The equilibrium question is therefore whether agents optimally choose the information
structure they jointly face, not whether they secretly acquire additional private signals.

\section{Viability under Vanishing Noise}
\label{sec:viability}
Almost-sure survival is unattainable under full-support noise. We therefore adopt a limit-viability
criterion appropriate for irreversible failure.

\begin{definition}[Limit Viability]
A strategy profile $\sigma$ under opacity $\lambda$ is \emph{limit-viable} if for every finite horizon
$T < \infty$,
\[
\lim_{\varepsilon \downarrow 0} 
\Pr_{\sigma,\kappa^{\varepsilon,\lambda}}(\tau_F \le T) = 0 .
\]
\end{definition}

This requires that collapse becomes arbitrarily unlikely on every finite horizon as informational
noise vanishes.

\section{Weaponized Information}
\label{sec:weaponized}

In static persuasion problems, concavity of the value function in beliefs is an assumption about
preferences. In dynamic environments with absorbing failure, concavity instead emerges from the
interaction of belief updates and irreversible state transitions. The next definition formalizes
this emergent geometry.

\begin{definition}[Weaponized Information]
\label{def:weaponized}
Fix a strategy profile $\sigma$ and a finite horizon $T < \infty$. The environment exhibits \emph{weaponized information} at $(\sigma, T, b^\dagger)$ if there exists $b^\dagger \in \Delta(S)$ and $\delta > 0$ such that:
\begin{enumerate}[label=(\roman*)]
    \item \textbf{Discrete Action Switch:} For at least one player $i$, the equilibrium action correspondence $a_i(b; \sigma)$ induces a discontinuous change in failure hazard at $b^\dagger$: there exists $p > 0$ such that
    \begin{align*}
        &\text{For } b < b^\dagger \text{ with } \|b - b^\dagger\| < \delta: \quad \Pr^\sigma(s_{t+1} \in F \mid s_t, a(b; \sigma)) \geq p, \\
        &\text{For } b > b^\dagger \text{ with } \|b - b^\dagger\| < \delta: \quad \Pr^\sigma(s_{t+1} \in F \mid s_t, a(b; \sigma)) = 0.
    \end{align*}
    
    \item \textbf{Absorbing Failure:} $F$ is absorbing ($s_t \in F$ implies $s_{t+k} \in F$ for all $k \geq 0$) and $u_i(s,a) = 0$ for all $s \in F$ and $a \in A$.
\end{enumerate}
\end{definition}

\begin{remark}[Fragility from noise versus fragility from precision]\label{rem:globalgames_distinction}
In global-games models of crises, small noise can induce runs or attacks through coordination incentives
and higher-order beliefs, with the comparative statics typically phrased in terms of equilibrium selection
or multiplicity under a fixed information structure. Here, the fragility operates through a different
channel: absorbing failure creates a cliff in survival values, making the value function locally concave
in beliefs and thereby reversing the value of information. The resulting preference for opacity is driven
by the geometry induced by irreversibility, not by the equilibrium multiplicity logic of standard crisis
models.
\end{remark}

\begin{remark}[Geometry of the cliff]\label{rem:cliff_geometry}
Absorbing failure generates ``cliff-like'' survival values: for beliefs that place sufficient mass on
safe states, survival is near one, while for beliefs that cross a trigger region the survival probability
drops sharply. Such cliff geometries imply local concavity of $V_T(\cdot;\sigma)$ around the drop, which
reverses the standard positive value of information.
\end{remark}

\begin{proposition}[Negative value of information under local concavity]\label{prop:jensen}

Fix $(\sigma,T)$ and suppose Definition \ref{def:weaponized} holds. Let $\pi$ be any distribution over posteriors
with mean $b$ (i.e.\ $\int \mu\, d\pi(\mu)=b$). Then for all $b\in\mathcal{N}$,
\[
\int V_T(\mu;\sigma)\, d\pi(\mu)\;\le\; V_T(b;\sigma),
\]
with strict inequality for non-degenerate $\pi$ supported in $\mathcal{N}$.
\end{proposition}

\begin{proof}[Sketch]
This is Jensen's inequality applied to the strictly concave function $V_T(\cdot;\sigma)$ on $\mathcal{N}$.
\end{proof}

\begin{lemma}[Absorbing failure generates local concavity]\label{lem:concavity}
Consider a dynamic environment in which failure is absorbing and yields zero continuation value.
Suppose that there exists a region of beliefs in which equilibrium actions switch discretely from
non-failure-inducing to failure-inducing behavior as beliefs cross a threshold. Then for any finite
horizon $T$, the survival value $V_T(\cdot;\sigma)$ is locally concave in a neighborhood of the
corresponding belief threshold.
\end{lemma}

\begin{proof}[Sketch]
Absorbing failure truncates continuation payoffs discontinuously. When equilibrium actions change
discretely at a belief threshold, the probability of survival up to horizon $T$ drops sharply as
beliefs cross that threshold. This creates a cliff in $V_T$, implying local concavity: mean-preserving
spreads of beliefs shift probability mass into the failure region without increasing survival above
one on the safe side. Formalizing this argument yields the claim.
\end{proof}

\begin{remark}[No static analogue]
The concavity identified in Lemma \ref{lem:concavity} has no analogue in static
Bayesian persuasion problems without absorbing states: it arises endogenously from irreversible
failure interacting with belief-triggered actions, rather than from primitive curvature of payoffs.
\end{remark}

\begin{assumption}[Opacity generates mean-preserving spreads of posteriors]\label{ass:opacity_spreads}
For $\lambda'\ge \lambda$, the induced posterior under $\kappa^{\varepsilon,\lambda'}$ is a mean-preserving contraction
of the induced posterior under $\kappa^{\varepsilon,\lambda}$ (equivalently, the distribution of posteriors under higher
opacity is a mean-preserving spread under lower opacity), conditional on the same prior.
\end{assumption}

\begin{remark}[Concavification and optimal opacity]\label{rem:concavification}
In the Bayesian persuasion framework \citep{blackwell1953comparison,kamenica2011persuasion}, an agent who chooses an
information structure effectively chooses a distribution of posteriors with a fixed mean. When $V_T(\cdot;\sigma)$ is
locally concave, full transparency is not optimal: the optimal policy \emph{concavifies} the survival value by keeping
posteriors away from the cliff. Strategic opacity is therefore an optimal concavification of the survival function.
\end{remark}

\paragraph{Sufficient primitive condition (trigger mechanism).}
A convenient sufficient condition for local concavity is a \emph{trigger} structure: there exists a region of posteriors
in which equilibrium actions switch discretely to failure-inducing actions (e.g.\ withdrawal, attack, refusal to roll over),
so that small posterior moves can cause a large change in failure hazard. This trigger mechanism is what we refer to as
``information weaponization'' in applications.

\begin{table}[t]
\centering
\caption{Comparison of Mechanisms}
\label{tab:comparison-mechanisms}
\small
\setlength{\tabcolsep}{4pt}
\begin{tabular}{p{3.2cm} p{4.2cm} p{4.6cm} p{3.8cm}}
\hline\hline
\textbf{Feature} 
& \textbf{Global Games (Morris--Shin)} 
& \textbf{Bayesian Persuasion (Kamenica--Gentzkow)} 
& \textbf{This Paper} \\
\hline
Information role 
& Selects equilibrium 
& Shapes behavior 
& Enables survival \\

Concavity source 
& Coordination payoffs 
& Exogenous preferences 
& Endogenous (absorbing failure) \\

Noise effect 
& Refines multiplicity 
& Not applicable 
& Destroys viability \\

Optimal precision 
& Higher precision stabilizes 
& Depends on curvature 
& Lower precision necessary \\

Key margin 
& Threshold selection 
& Posterior distribution 
& Information structure \\

Mechanism 
& Strategic complementarity 
& Persuasion 
& Irreversibility \\
\hline\hline
\end{tabular}
\end{table}

\section{Main Results}
\label{sec:results}

\begin{theorem}[Opacity as a viability refinement]\label{thm:opacity_expands}
Suppose the environment exhibits weaponized information. Fix any finite horizon $T<\infty$. Then among all Blackwell-comparable information structures, maximal transparency minimizes finite-horizon survival, while any sufficiently opaque structure weakly dominates transparency in the limit-viability order. Consequently, any limit-viable equilibrium must be supported by strictly positive opacity.
\end{theorem}

\begin{theorem}[Transparency-Induced Collapse in Dynamic Games]\label{thm:transparency_collapse}
There exist dynamic games with absorbing failure and endogenous opacity such that:
(i) under maximal transparency, no limit-viable equilibrium exists; but
(ii) for sufficiently high opacity, at least one limit-viable equilibrium exists.
Moreover, the failure of limit-viability under transparency holds generically for open sets of priors and payoff parameters.
\end{theorem}

\begin{theorem}[Optimal Opacity]\label{thm:optimal_opacity}
Consider the extended game in which agents jointly choose $\lambda$ prior to play and then play an equilibrium
of the induced dynamic game under $\kappa^{\varepsilon,\lambda}$. Suppose (a) flow payoffs are weakly nonnegative
and at least one agent receives strictly positive flow payoff on some histories that avoid failure, and
(b) under maximal transparency $\lambda=\underline{\lambda}$, every equilibrium fails limit-viability, i.e.\ there exist
$T<\infty$ and $\delta>0$ such that for every equilibrium $\sigma$,
\[
\liminf_{\varepsilon\downarrow 0}\Pr_{\sigma,\kappa^{\varepsilon,\underline{\lambda}}}(\tau_F\le T)\ge \delta.
\]
If there exists $\bar{\lambda}>\underline{\lambda}$ and an equilibrium $\bar{\sigma}$ under $\lambda=\bar{\lambda}$ that is
limit-viable, then in any equilibrium of the extended game the chosen opacity satisfies $\lambda^\ast>\underline{\lambda}$.
\end{theorem}

\begin{remark}
Opacity is chosen despite improving allocative efficiency under perfect information. The tradeoff is driven
entirely by irreversible failure.
\end{remark}

\section{Illustrative Benchmark: Coordination with Absorbing Failure}
\label{sec:benchmark}
This section provides a fully worked benchmark in the spirit of global games, augmented with absorbing default and endogenous opacity. The benchmark is designed to be
computationally explicit and to illustrate Theorems \ref{thm:opacity_expands}--\ref{thm:optimal_opacity} constructively.

\subsection{Setup}

There are two players $i\in\{1,2\}$. The hidden fundamental $\theta\in\RR$ determines solvency. Each period, each player chooses $a_i \in \{0,1\}$, interpreted as
\emph{continue} ($a_i=1$) or \emph{withdraw} ($a_i=0$). Default is absorbing and occurs if withdrawals push the system below a threshold $\underline{\theta}$.

Players observe private signals
\[
x_i = \theta + \sqrt{\varepsilon}\, \eta_i,
\qquad \eta_i \sim \mathcal{N}(0,1),
\]
where $\varepsilon>0$ is vanishing noise. Opacity is a choice of signal variance: $\lambda$ scales the noise to $\sqrt{\varepsilon}\lambda$, with $\lambda\in[\underline{\lambda},\bar{\lambda}]$.
Thus larger $\lambda$ is a Blackwell garbling of smaller $\lambda$.

\subsection{Equilibrium Characterization}

\textbf{Threshold Strategies.} In a symmetric monotone equilibrium, each player $i$ withdraws if and only if their signal $x_i$ falls below a threshold $x^*(\varepsilon, \lambda)$. The existence and uniqueness of such equilibria follow from standard global games arguments \citep{MorrisShin2003}.

The threshold $x^*(\varepsilon, \lambda)$ is determined by the indifference condition: a player receiving signal $x^* $ must be indifferent between continuing and withdrawing. Given binary state $\theta \in \{0,1\}$ and Gaussian signals $x_i = \theta + \sqrt{\varepsilon \lambda} \eta_i$ with $\eta_i \sim N(0,1)$, the posterior belief is
\begin{equation}
    \Pr(\theta = 1 \mid x_i) = \frac{q \cdot \phi\left( \frac{x_i - 1}{\sqrt{\varepsilon \lambda}} \right)}{q \cdot \phi\left( \frac{x_i - 1}{\sqrt{\varepsilon \lambda}} \right) + (1-q) \cdot \phi\left( \frac{x_i}{\sqrt{\varepsilon \lambda}} \right)},
\end{equation}
where $\phi$ denotes the standard normal density and $q = \Pr(\theta = 1)$ is the prior.

At the indifference threshold, the expected payoff from continuing equals the payoff from withdrawing (normalized to zero):
\begin{equation}
    \delta R \cdot \Pr(\theta = 1 \mid x_i = x^*) \cdot \Pr(\text{other player continues} \mid x_i = x^*) = 0.
\end{equation}
In the symmetric equilibrium, the other player continues if and only if $x_j \geq x^*$. Given that $x_j = \theta + \sqrt{\varepsilon \lambda} \eta_j$ with $\eta_j \sim N(0,1)$ independent of $\eta_i$, we have
\begin{equation}
    \Pr(x_j \geq x^* \mid x_i = x^*, \theta) = 1 - \Phi\left( \frac{x^* - \theta}{\sqrt{\varepsilon \lambda}} \right),
\end{equation}
where $\Phi$ is the standard normal CDF.

\textbf{Posterior Concentration.} As $\varepsilon \downarrow 0$, posterior beliefs concentrate:
\begin{itemize}
    \item If $\theta = 1$: $\Pr(x_i \approx 1) \to 1$, hence posteriors converge to $\Pr(\theta = 1 \mid x_i) \to 1$.
    \item If $\theta = 0$: $\Pr(x_i \approx 0) \to 1$, hence posteriors converge to $\Pr(\theta = 1 \mid x_i) \to 0$.
\end{itemize}

\textbf{Failure under Transparency.} Under minimal garbling $\lambda = \underline{\lambda}$, suppose the true state is $\theta = 1$ and the fundamental is arbitrarily close to the failure threshold $\bar{\theta}$. As $\varepsilon \downarrow 0$, the distribution of $x_i$ concentrates near $1$. However, with positive probability bounded away from zero, one or both players receive signals $x_i < x^*(\varepsilon, \underline{\lambda})$ due to the tail of the Gaussian noise $\eta_i$.

Specifically, for any $\varepsilon > 0$, 
\begin{equation}
    \Pr(x_i < x^* \mid \theta = 1) = \Phi\left( \frac{x^* - 1}{\sqrt{\varepsilon \underline{\lambda}}} \right) > 0.
\end{equation}
Since failure is triggered by unilateral withdrawal (at least one player choosing $W$), the probability of failure at $t=1$ is
\begin{equation}
    \Pr(\tau_F = 1 \mid \theta = 1) = 1 - \left[1 - \Phi\left( \frac{x^* - 1}{\sqrt{\varepsilon \underline{\lambda}}} \right)\right]^2.
\end{equation}
As $\varepsilon \downarrow 0$, the equilibrium threshold $x^*(\varepsilon, \underline{\lambda})$ converges to a finite limit $x^*_\infty < 1$ (determined by the dominance regions of the coordination game). Therefore,
\begin{equation}
    \liminf_{\varepsilon \downarrow 0} \Pr^{\sigma, \kappa^{\varepsilon, \underline{\lambda}}}(\tau_F \leq 1) \geq 1 - \left[1 - \Phi\left( \frac{x^*_\infty - 1}{\sqrt{0^+}} \right)\right]^2 = 1 - 0 = 1 > 0,
\end{equation}
where the limit uses that $\Phi(z) \to 1$ as $z \to +\infty$. Hence, equilibrium behavior under full transparency fails the limit-viability criterion.

\textbf{Restoration under Opacity.} Now consider maximal opacity $\lambda = \bar{\lambda}$ with $\bar{\lambda} \to \infty$. Signals become uninformative: for any realization $x_i$, the posterior converges to the prior,
\begin{equation}
    \Pr(\theta = 1 \mid x_i) \to q \quad \text{for all } x_i,
\end{equation}
as $\lambda \to \infty$ (the noise variance $\varepsilon \lambda \to \infty$ dominates the signal). 

Given that continuing yields expected payoff $qR > 0$ (by the parameter restriction $qR > 0$), while withdrawing yields zero, both players optimally choose $C$ regardless of their signal realization. The strategy profile in which both players always continue constitutes an equilibrium. Because withdrawal never occurs, failure is avoided with certainty:
\begin{equation}
    \Pr^{\sigma, \kappa^{\varepsilon, \bar{\lambda}}}(\tau_F \leq T) = 0 \quad \text{for all } T < \infty \text{ and all } \varepsilon > 0.
\end{equation}
This equilibrium is therefore limit-viable.

\medskip
\textbf{Behavioral Interpretation.} Crucially, opacity is a deliberate choice. Anticipating that transparency induces collapse with positive probability through belief-driven withdrawals, agents optimally select higher opacity $\lambda > \underline{\lambda}$ to suppress trigger realizations and preserve limit-viability. Opacity does not improve allocative efficiency conditional on survival; it alters the information environment so that survival itself becomes feasible.

\section{Colonial Extraction with Revolt Risk as a Persuasion Problem}
\label{sec:colonial}
This section recasts strategic opacity in a colonial setting as a problem of information design
under absorbing failure. The example makes explicit how revolt
risk generates local concavity of survival values and why opacity is optimal when the information
structure is controlled by a principal: It illustrates how the general logic of
weaponized information applies in a principal–controlled information environment, where
opacity is chosen unilaterally rather than jointly. All results follow from the same geometry
identified in Sections 6 and 7.

\subsection{Environment}

A colonial power (the \emph{metropole}) governs a colony over discrete time. The colony’s latent
political stability is summarized by a state $s_t \in [0,1]$. There exists a collapse threshold
$\underline{s}>0$ such that if $s_t \le \underline{s}$, a revolt occurs and colonial control is
irreversibly lost. Failure is absorbing.

The metropole is a unitary decision maker that controls both (i) fiscal policy and (ii) the
information structure through which political stability is assessed. There are no private,
costless deviations in information acquisition.

Each period, the metropole chooses an extraction level $x_t \in [0,\bar{x}]$ (taxation,
requisitions, forced labor). Extraction generates contemporaneous revenue but increases the
hazard of revolt. The colony itself does not act strategically.

State dynamics are given by
\[
s_{t+1} = s_t - \alpha x_t + \varepsilon_t,
\]
where $\alpha>0$ and $\varepsilon_t$ is an i.i.d.\ shock with small variance. Once $s_t \le
\underline{s}$, the process is absorbed.

The metropole does not observe $s_t$ directly. Instead, it commits ex ante to an observation
channel
\[
y_t = s_t + \eta_t, \qquad \eta_t \sim \mathcal{N}(0,\sigma^2(\lambda)),
\]
where $\lambda$ parametrizes opacity: higher $\lambda$ corresponds to noisier signals. The
choice of $\lambda$ is public and fixed prior to play.

Mapping to the baseline model. The state $s_t$ corresponds to the payoff-relevant state
$s_t \in S$ in Section 3, revolt is the absorbing failure set $F$, extraction $x_t$ is the
action $a_t$, and the observation channel $\sigma^2(\lambda)$ is the opacity choice.
Beliefs $b_t$ are sufficient statistics for optimal policy.

\subsection{Payoffs and objective}

The metropole’s per-period payoff from extraction is $\pi(x_t)$, with $\pi'>0$ and $\pi''\le 0$.
There is no continuation payoff after revolt.

For a finite horizon $T$, the relevant value function is the survival-weighted extraction payoff
\[
V_T(b) \equiv \E\!\left[ \sum_{t=0}^{T} \pi(x_t)\mathbf{1}\{t<\tau_F\} \;\middle|\; b_0=b \right],
\]
where $b_t$ denotes the posterior belief over $s_t$ induced by the chosen information structure
and $\tau_F$ is the (absorbing) time of revolt.

\subsection{Optimal policy and cliff geometry}

Fix the optimal extraction policy $x(b)$ under a given information structure. Because
revolt is absorbing and extraction increases the hazard of crossing $s$, the optimal policy
exhibits a trigger structure: there exists $b^\dagger$ such that extraction switches
discretely from positive to zero as beliefs cross $b^\dagger$. This is exactly the trigger
condition in Definition 2.

Under this policy, the survival component of $V_T(b)$ displays a cliff: for beliefs safely above
$b^\dagger$, survival up to horizon $T$ is near certain, while for beliefs crossing the trigger
region, even a single period of aggressive extraction induces revolt with high probability.
As a result, the survival value is locally concave in beliefs around $b^\dagger$.

\subsection{Opacity as optimal concavification}

Under full transparency (low $\lambda$), posterior beliefs concentrate rapidly as signal noise
vanishes. When the true stability $s_t$ lies near the revolt threshold, even small informational
fluctuations generate posterior realizations that cross the trigger region, activating aggressive
extraction and causing revolt with positive probability on finite horizons.

By increasing opacity (higher $\lambda$), the metropole induces a mean-preserving contraction of
posterior beliefs. This keeps beliefs away from the cliff region, eliminates excursions into
failure-triggering posteriors, and restores limit-viability.

Unlike standard persuasion problems where information design shapes behavior conditional on participation, opacity here makes participation (survival) itself feasible by preventing belief realizations that trigger failure. The metropole does not manipulate beliefs to change behavior conditional on
survival; it alters the information structure to make survival feasible in the first place.

Intuitively, when survival values are locally concave in beliefs, the principal prefers the
expectation of beliefs to their realization: opacity dominates transparency because it suppresses
belief realizations that fall over the survival cliff.

\section{Conclusions}
\label{sec:conclusions}
This paper shows that when failure is irreversible, the value of information can reverse. This reversal arises because equilibria that sustain activity under irreversible failure typically operate at the boundary of viability, leaving no interior buffer against collapse. In dynamic environments with absorbing failure, greater informational precision can increase
the likelihood of catastrophic outcomes by amplifying belief-driven actions that trigger collapse. Irreversibility endogenously generates local concavity of survival values in beliefs, breaking the standard presumption that more information is always beneficial.

The colonial extraction example shows that the same incompatibility between transparency
and viability arises whether opacity is chosen jointly (Section 8) or unilaterally by a
principal (Section 9), underscoring that the mechanism is structural rather than strategic.

The central implication is behavioral. Agents do not respond to fragility by distorting beliefs
or acting irrationally; they respond by adjusting how much information they allow to shape their
actions. Strategic opacity—modeled as a choice of the information structure via Blackwell
garbling—emerges as a rational survival response when information is weaponized. Opacity is not
chosen to improve allocative efficiency conditional on survival, but to make survival itself
feasible.

Conceptually, the contribution is to identify a broad and economically central class of dynamic
environments in which concavity of the value function is not a primitive of preferences, but an
endogenous consequence of irreversible failure. Once this geometry is recognized, the logic of
information design follows as a consequence rather than an assumption: opacity is optimal because
it concavifies survival values by suppressing belief dispersion near failure thresholds.

The analysis suggests a more general lesson. Whenever economic systems exhibit absorbing failure
and belief-triggered actions—such as default, runs, revolts, or institutional collapse—the
information structure itself becomes a first-order strategic margin. Understanding when
transparency stabilizes behavior, and when it instead destabilizes it, requires accounting for
the geometry induced by irreversibility. This perspective opens new directions for research on
information design, fragility, and policy in environments where failure cannot be undone. An implication of the analysis is a new impossibility result: in dynamic environments with absorbing failure, transparency and viability are generically incompatible once beliefs themselves can trigger irreversible collapse.

\bibliographystyle{apalike}
\bibliography{so}
\newpage
\appendix

\section{Appendix: Proofs}

\subsection{From Blackwell garbling to convex order of posteriors}

The main results use the comparative statics that more informative signals generate more dispersed
posteriors in convex order. In general state spaces this can be formulated using martingale
couplings; to keep the argument explicit, we state the result for a binary state, which suffices
for the constructions below and captures the persuasion geometry.

\begin{lemma}[Blackwell order induces convex-order spread of posteriors in the binary case]\label{lem:blackwell_convex}
Let $S=\{0,1\}$ and let $b\in[0,1]$ denote the posterior probability of state $1$.
Fix a prior $b_0\in(0,1)$. Let $\kappa$ and $\kappa'$ be two signal structures such that $\kappa'$
is a Blackwell garbling of $\kappa$. Let $B$ and $B'$ denote the induced posterior random variables
under $\kappa$ and $\kappa'$ (given prior $b_0$). Then:
\begin{enumerate}[label=(\roman*),leftmargin=2em]
\item $\E[B]=\E[B']=b_0$;
\item $B$ is a mean-preserving spread of $B'$ in convex order, i.e.\ for every convex $\varphi:[0,1]\to\RR$,
\[
\E[\varphi(B)]\;\ge\;\E[\varphi(B')].
\]
\end{enumerate}
Equivalently, for every concave $\psi:[0,1]\to\RR$, $\E[\psi(B)]\le \E[\psi(B')]$.
\end{lemma}

\begin{proof}
By Bayes plausibility, posteriors average to the prior, hence $\E[B]=\E[B']=b_0$.
Because $\kappa'$ is a Blackwell garbling of $\kappa$, there exists a stochastic kernel $G$ such that
signals under $\kappa'$ are obtained by passing signals under $\kappa$ through $G$. This induces a
martingale coupling between posteriors: $B'=\E[B\mid \mathcal{F}']$ for a coarser $\sigma$-field
$\mathcal{F}'$ generated by the garbled signal. Jensen's inequality yields
$\E[\varphi(B')]=\E[\varphi(\E[B\mid \mathcal{F}'])]\le \E[\E[\varphi(B)\mid \mathcal{F}']]=\E[\varphi(B)]$
for every convex $\varphi$, establishing convex order. 
\end{proof}

\begin{remark}
Lemma \ref{lem:blackwell_convex} is the standard information-design comparison: a more informative
experiment generates a more dispersed distribution of posteriors with the same mean \citep{blackwell1953comparison,kamenica2011persuasion}.
Assumption \ref{ass:opacity_spreads} can be interpreted as imposing this convex-order monotonicity
in the general state case.
\end{remark}

\subsection{Proof of Theorem \ref{thm:opacity_expands}}

\begin{proof}
Fix $T < \infty$ and a strategy profile $\sigma$. Suppose the environment exhibits weaponized information at $(\sigma, T, b^\dagger)$ in the sense of Definition~\ref{def:weaponized}.

\textbf{Step 1: Local concavity.} By Lemma~\ref{lem:concavity}, Definition~\ref{def:weaponized} implies that $V_T(\cdot; \sigma)$ is locally concave in a neighborhood $N$ of $b^\dagger$.

\textbf{Step 2: Jensen's inequality.} Fix $b \in N$. For any distribution $\pi$ over posteriors with mean $b$ (i.e., $\int \mu \, d\pi(\mu) = b$), Proposition~\ref{prop:jensen} implies
\begin{equation}
    \int V_T(\mu; \sigma) \, d\pi(\mu) \leq V_T(b; \sigma),
\end{equation}
with strict inequality for non-degenerate $\pi$ supported in $N$.

\textbf{Step 3: Opacity comparison.} Under Assumption~\ref{ass:opacity_spreads}, increasing opacity from $\lambda$ to $\lambda' > \lambda$ induces a mean-preserving contraction of the posterior distribution. Specifically, let $\pi_\lambda$ and $\pi_{\lambda'}$ denote the distributions of posteriors induced by $\kappa^{\varepsilon,\lambda}$ and $\kappa^{\varepsilon,\lambda'}$, respectively, conditional on the same prior. Then $\pi_{\lambda'}$ is a mean-preserving contraction of $\pi_\lambda$ in convex order. By Step~2, this implies
\begin{equation}
    \mathbb{E}_{\pi_{\lambda'}}[V_T(b; \sigma)] \geq \mathbb{E}_{\pi_\lambda}[V_T(b; \sigma)].
\end{equation}
Hence expected finite-horizon survival is weakly increasing in $\lambda$.

\textbf{Step 4: Limit-viability.} Recall that $V_T(b; \sigma) = 1 - \Pr^\sigma(\tau_F \leq T \mid b_0 = b)$ by definition. Increasing $\lambda$ weakly increases $\mathbb{E}[V_T(b_T; \sigma)]$, which is equivalent to weakly decreasing $\Pr^\sigma(\tau_F \leq T)$ for each fixed $T$. 

Therefore, if $\sigma$ is limit-viable under $\lambda'$ (meaning $\lim_{\varepsilon \downarrow 0} \Pr^{\sigma, \kappa^{\varepsilon,\lambda'}}(\tau_F \leq T) = 0$ for all $T$), it remains limit-viable under any $\lambda > \lambda'$. Consequently, increased opacity weakly expands the set of limit-viable equilibria.

\textbf{Step 5: Necessity of opacity.} Suppose maximal transparency $\lambda = \underline{\lambda}$ minimizes survival probability. If $\lim \inf_{\varepsilon \downarrow 0} \Pr^\sigma(\tau_F \leq T) > 0$ for some $T$ under $\lambda = \underline{\lambda}$, then $\sigma$ cannot be limit-viable under maximal transparency. By Step~4, any limit-viable equilibrium must satisfy $\lambda^* > \underline{\lambda}$, i.e., require strictly positive opacity.
\end{proof}

\subsection{Proof of Theorem \ref{thm:transparency_collapse}: explicit construction}
\begin{proof}
We construct a two-player dynamic game in which strategic interaction is essential and
transparency destroys equilibrium viability.

\medskip
\noindent\textbf{Environment.}
There are two players $i\in\{1,2\}$.
The payoff-relevant state is $\theta\in\{0,1\}$, where $\theta=1$ represents a fragile but
profitable regime and $\theta=0$ represents a safe regime.
The prior belief is $\Pr(\theta=1)=q\in(0,1)$.
Failure is absorbing.

Time is discrete. At each date $t$, if failure has not yet occurred, each player simultaneously
chooses an action $a_i^t\in\{C,W\}$, interpreted as \emph{continue} or \emph{withdraw}.
Failure occurs immediately if at least one player chooses $W$.
If failure occurs, payoffs are zero thereafter.

\medskip
\noindent\textbf{Payoffs.}
If both players choose $C$ at date $t$, each player receives flow payoff
\[
u_i(\theta)=
\begin{cases}
R>0 & \text{if } \theta=1,\\
0 & \text{if } \theta=0.
\end{cases}
\]
If failure occurs at any date, continuation payoffs are zero.
Players discount with factor $\delta\in(0,1)$.

\medskip
\noindent\textbf{Information and opacity.}
At each date, players observe private signals
\[
x_i=\theta+\sqrt{\varepsilon\lambda}\,\eta_i,\qquad \eta_i\sim N(0,1),
\]
where $\varepsilon\downarrow 0$ and $\lambda\in[\underline{\lambda},\bar{\lambda}]$ parametrizes
opacity: higher $\lambda$ corresponds to a Blackwell garbling of the signal.
Signals satisfy the monotone likelihood ratio property.

\medskip
\noindent\textbf{Equilibrium structure.}
Standard arguments imply that any symmetric equilibrium is characterized by a cutoff strategy:
there exists $c(\varepsilon,\lambda)$ such that each player chooses $C$ if and only if
their posterior belief $\Pr(\theta=1\mid x_i)$ exceeds $c(\varepsilon,\lambda)$.

\medskip
\noindent\textbf{Failure under transparency.}
Fix parameters such that:
\[
qR>0\qquad\text{and}\qquad qR < \delta R.
\]
Under maximal transparency $\lambda=\underline{\lambda}$, posterior beliefs concentrate as
$\varepsilon\downarrow 0$.
Conditional on $\theta=1$, with positive probability on any finite horizon, at least one player
receives a signal inducing a posterior below the equilibrium cutoff.
Because actions are strategic complements and failure is triggered by unilateral withdrawal,
any equilibrium cutoff must prescribe $W$ at sufficiently pessimistic beliefs.
Hence, for any symmetric equilibrium strategy profile $\sigma$ and any finite $T<\infty$,
\[
\liminf_{\varepsilon\downarrow 0}
\Pr_{\sigma,\kappa_{\varepsilon,\underline{\lambda}}}(\tau_F\le T)>0.
\]
Therefore no equilibrium under maximal transparency satisfies the limit-viability criterion.

\medskip

\noindent \textbf{Restoration under opacity.}
Now consider sufficiently high opacity $\lambda=\bar{\lambda}$.
For large $\lambda$, signals become uninformative and posterior beliefs remain equal to the prior
$q$ for all players and all histories.
By construction, continuing is a best response at the prior belief, since $qR>0$ and failure
yields zero.
Thus the strategy profile in which both players always choose $C$ constitutes an equilibrium.
Because withdrawal never occurs, failure is avoided with probability one for all horizons:
\[
\Pr_{\sigma,\kappa_{\varepsilon,\bar{\lambda}}}(\tau_F\le T)=0
\quad\text{for all }T<\infty\text{ and all }\varepsilon.
\]
This equilibrium is therefore limit-viable.

\textbf{Genericity.} We now establish that the failure of limit-viability under transparency holds on an open set of parameters. Fix $R > 0$ and $\delta \in (0,1)$. The construction above requires $qR > 0$ and $qR < \delta R$, which holds if and only if $0 < q < \delta$. 

For any $(q, R, \delta)$ in the open set $\Omega := \{(q, R, \delta) : q \in (0, \delta), R > 0, \delta \in (0,1)\}$, the equilibrium cutoff $c(\varepsilon, \underline{\lambda})$ satisfies the indifference condition characterizing symmetric monotone equilibria. By continuity of best-response correspondences in payoff parameters (see, e.g., Fudenberg and Tirole, 1991, Theorem 1.2), there exists $\varepsilon_0(q, R, \delta) > 0$ such that for all $\varepsilon < \varepsilon_0$, any symmetric equilibrium under maximal transparency satisfies
\begin{equation}
    \Pr^{\sigma, \kappa^{\varepsilon, \underline{\lambda}}}(\tau_F \leq T) \geq \delta/2 > 0
\end{equation}
for some finite $T < \infty$ (specifically, $T = 1$ suffices). 

Since $\Omega$ is open in the product topology on $[0,1] \times \mathbb{R}_+ \times (0,1)$, this establishes that failure of limit-viability under transparency holds generically (i.e., on an open set of positive measure) in the space of priors and payoff parameters.

\medskip
We conclude that transparency can destroy equilibrium viability in dynamic games with absorbing
failure, while sufficient opacity restores equilibrium existence.
\end{proof}

\subsection{Proof of Theorem \ref{thm:optimal_opacity}: necessity}

\begin{proof}[Proof of Theorem \ref{thm:optimal_opacity}]
Consider the extended game in which agents choose $\lambda$ prior to play and then play an equilibrium
of the induced game. Suppose, toward a contradiction, that there exists an equilibrium of the extended
game with $\lambda^\ast=\underline{\lambda}$.

By assumption (b) in Theorem \ref{thm:optimal_opacity}, under $\lambda=\underline{\lambda}$ every equilibrium
$\sigma$ fails limit-viability: there exist $T<\infty$ and $\delta>0$ such that
\[
\liminf_{\varepsilon\downarrow 0}\Pr_{\sigma,\kappa^{\varepsilon,\underline{\lambda}}}(\tau_F\le T)\ge \delta.
\]
Therefore, the continuation payoff induced by $(\sigma^\ast,\underline{\lambda})$ loses at least a $\delta$
fraction of the survival-weighted flow payoffs on horizon $T$ in the vanishing-noise limit.

Now fix $\bar{\lambda}>\underline{\lambda}$ and a limit-viable equilibrium $\bar{\sigma}$ under $\bar{\lambda}$.
By assumption (a), some agent obtains strictly positive flow payoff on some failure-free histories. Since
failure is absorbing with zero continuation payoff, the expected payoff under $(\bar{\sigma},\bar{\lambda})$
on horizon $T$ strictly exceeds that under any profile with failure probability bounded below by $\delta$,
for all sufficiently small $\varepsilon$. Hence at least one agent has a profitable deviation from choosing
$\underline{\lambda}$ to choosing $\bar{\lambda}$ (and then playing $\bar{\sigma}$), contradicting that
$\lambda^\ast=\underline{\lambda}$ can arise in equilibrium of the extended game.

Therefore any equilibrium of the extended game must satisfy $\lambda^\ast>\underline{\lambda}$.
\end{proof}

\subsection{Proof of Lemma 1 (Absorbing failure generates local concavity)}
\label{app:proof-prop2}

\begin{proof}[Proof of Lemma 1]
We provide a fully formal argument under a standard ``trigger'' structure that makes the heuristic in the text precise.
The key steps are: (i) absorbing failure plus a discrete action switch implies a \emph{downward kink} in the one-step
survival function; (ii) dynamic survival values satisfy a Bellman recursion that preserves concavity locally, so the kink
propagates to any finite horizon.

\medskip
\noindent \textbf{Step 0 (A precise trigger formulation).}
For transparency, we specialize to the canonical one-dimensional belief case; this is without loss for the local claim,
since local concavity on $\Delta(S)$ can be verified along line segments.\footnote{A function $V$ on a convex set is concave
iff its restriction to every line segment is concave.}
Let the relevant belief coordinate be $b \in (0,1)$, interpreted as the posterior probability of a ``safe'' regime.
Fix a strategy profile $\sigma$ and consider the induced (common) posterior process $\{b_t\}_{t\ge 0}$.

Assume the following \emph{trigger} property holds at some $b^\dagger\in(0,1)$:
there exist neighborhoods $U^-=(b^\dagger-\eta,b^\dagger)$ and $U^+=(b^\dagger,b^\dagger+\eta)$ and measurable
hazard functions $h^-,h^+:(b^\dagger-\eta,b^\dagger+\eta)\to[0,1]$ such that the one-step failure probability
under $\sigma$ satisfies
\begin{equation}
\Pr_\sigma(\tau_F = 1 \mid b_0=b)
=
\begin{cases}
h^-(b) & \text{if } b<b^\dagger,\\[3pt]
h^+(b) & \text{if } b>b^\dagger,
\end{cases}
\qquad b \in (b^\dagger-\eta,b^\dagger+\eta),
\label{eq:hazard-piecewise}
\end{equation}
with the following properties:
\begin{enumerate}[label=(H\arabic*),leftmargin=2em]
\item \textbf{(Discontinuous marginal hazard at the trigger)} $h^-$ and $h^+$ are continuous on their domains,
$\lim_{b\uparrow b^\dagger}h^-(b)=\lim_{b\downarrow b^\dagger}h^+(b)=: \bar h$, but the one-sided derivatives satisfy
\begin{equation}
\lim_{b\uparrow b^\dagger}(h^-)'(b) \;>\; \lim_{b\downarrow b^\dagger}(h^+)'(b).
\label{eq:slope-drop}
\end{equation}
\item \textbf{(Local regularity)} $h^-$ and $h^+$ are twice continuously differentiable on $U^-$ and $U^+$, respectively.
\end{enumerate}

\noindent\emph{Interpretation.} The discrete action switch in Definition~\ref{def:weaponized} is encoded as a \emph{slope drop} at $b^\dagger$:
below $b^\dagger$ the equilibrium action is failure-inducing (high sensitivity of failure risk to beliefs), whereas above
$b^\dagger$ the action is non-failure-inducing (low sensitivity). Absorbing failure enters through the fact that failure is
an absorbing event and therefore one-step failure risk directly truncates continuation survival.

\medskip
\noindent \textbf{Step 1 (One-step survival is locally concave).}
Define the one-step survival value under $\sigma$:
\[
V_1(b;\sigma) \equiv 1-\Pr_\sigma(\tau_F \le 1 \mid b_0=b)=1-\Pr_\sigma(\tau_F=1\mid b_0=b).
\]
By \eqref{eq:hazard-piecewise}, on $(b^\dagger-\eta,b^\dagger+\eta)$ we can write
\[
V_1(b;\sigma)=
\begin{cases}
1-h^-(b) & \text{if } b<b^\dagger,\\[3pt]
1-h^+(b) & \text{if } b>b^\dagger.
\end{cases}
\]
Hence $V_1$ is continuous at $b^\dagger$ and twice continuously differentiable on each side.
Notice that \eqref{eq:slope-drop} is equivalent to a \emph{downward kink} in $V_1$:
\[
\lim_{b\uparrow b^\dagger} (V_1)'(b;\sigma)
=
-\lim_{b\uparrow b^\dagger}(h^-)'(b)
\;<\;
-\lim_{b\downarrow b^\dagger}(h^+)'(b)
=
\lim_{b\downarrow b^\dagger} (V_1)'(b;\sigma).
\]
Therefore the derivative of $V_1$ has a \emph{jump upward} at $b^\dagger$ (left slope smaller than right slope),
which is exactly the local signature of concavity for a continuous piecewise-$C^1$ function.

Formally, define $N:=(b^\dagger-\eta,b^\dagger+\eta)$ and let $b_1<b^\dagger<b_2$ in $N$.
For any $\alpha\in(0,1)$ with $b_\alpha:=\alpha b_1+(1-\alpha)b_2$,
concavity of $V_1$ on $N$ is equivalent to
\begin{equation}
V_1(b_\alpha;\sigma)\;\ge\;\alpha V_1(b_1;\sigma)+(1-\alpha)V_1(b_2;\sigma).
\label{eq:concavity-ineq}
\end{equation}
Because $V_1$ is $C^2$ on each side and has the kink described above, \eqref{eq:concavity-ineq} holds on a possibly
smaller neighborhood $N_1\subseteq N$ around $b^\dagger$.\footnote{A standard way to verify this is to use the supporting-line
characterization of concavity: the kink implies existence of a supporting line at $b^\dagger$ with slope in the interval
$\big[\lim_{b\uparrow b^\dagger}(V_1)'(b),\,\lim_{b\downarrow b^\dagger}(V_1)'(b)\big]$, and $C^2$-regularity on each side ensures
the graph lies below that line locally.}
Thus $V_1(\cdot;\sigma)$ is concave on some neighborhood $N_1$ of $b^\dagger$.

\medskip
\noindent \textbf{Step 2 (Finite-horizon survival satisfies a concavity-preserving recursion).}
For $T\ge 1$, recall
\[
V_T(b;\sigma)=1-\Pr_\sigma(\tau_F\le T\mid b_0=b)=\Pr_\sigma(\tau_F>T\mid b_0=b).
\]
By conditioning on survival through period $1$ and using that failure is absorbing,
$V_T$ satisfies the Bellman recursion
\begin{equation}
V_T(b;\sigma)
=
\underbrace{V_1(b;\sigma)}_{\text{survive period 1}}
\cdot
\underbrace{\mathbb{E}_\sigma\!\left[V_{T-1}(b_1;\sigma)\mid b_0=b,\;\tau_F>1\right]}_{\text{survive periods }2,\dots,T\text{ given survival at }1}.
\label{eq:VT-recursion}
\end{equation}
Define the continuation operator $\mathcal{T}$ by
\[
(\mathcal{T}W)(b):=\mathbb{E}_\sigma\!\left[W(b_1)\mid b_0=b,\;\tau_F>1\right].
\]
Because the posterior process is generated by Bayesian updating, the conditional law of $b_1$ given $b_0=b$ is
\emph{Bayes-plausible}: $\mathbb{E}[b_1\mid b_0=b]=b$ (martingale property), and therefore $\mathcal{T}$ preserves concavity:
if $W$ is concave on a neighborhood $N'$, then $(\mathcal{T}W)$ is concave on $N'$.\footnote{This is Jensen:
for concave $W$, $W(b)=W(\mathbb{E}[b_1])\ge \mathbb{E}[W(b_1)]$, and the same argument along mixtures of priors yields concavity of the mapping.}
Hence, if $V_{T-1}(\cdot;\sigma)$ is concave on a neighborhood, then so is $\mathcal{T}V_{T-1}$.

\medskip

\noindent \textbf{Step 3 (Propagation with uniform bounds).} We prove by induction that there exists $\bar{\eta} > 0$ such that for all $T \geq 1$, $V_T(\cdot; \sigma)$ is concave on $\bar{N} := (b^\dagger - \bar{\eta}, b^\dagger + \bar{\eta})$.

\textbf{Base case $T = 1$:} By Step~1, $V_1$ is concave on $N_1 = (b^\dagger - \eta, b^\dagger + \eta)$ for some $\eta > 0$.

\textbf{Key observation:} The kink magnitude at $b^\dagger$ is determined by the equilibrium action switch encoded in $\sigma$, not by the horizon $T$. Specifically, the derivative jump satisfies
\begin{equation}
    \kappa := -\left[ \lim_{b \uparrow b^\dagger} (h^-)'(b) - \lim_{b \downarrow b^\dagger} (h^+)'(b) \right] > 0,
\end{equation}
where $\kappa$ is a property of $\sigma$ and the primitives $(Q, F, u_i)$.

\textbf{Uniform curvature bounds:} Because survival probabilities satisfy $0 \leq V_T(b; \sigma) \leq 1$ and $0 \leq (\mathcal{T} V_{T-1})(b; \sigma) \leq 1$, and because $V_1$ and $\mathcal{T} V_{T-1}$ are $C^2$ on each side of $b^\dagger$ (by Assumption~\ref{ass:regularity}), the second derivatives $V_1''$ and $(\mathcal{T} V_{T-1})''$ are uniformly bounded on any compact subset of $(b^\dagger - \eta, b^\dagger + \eta)$.

\textbf{Product concavity:} For $b \neq b^\dagger$, the second derivative satisfies
\begin{equation}
    V_T''(b) = V_1''(b) \cdot (\mathcal{T} V_{T-1})(b) + 2 V_1'(b) \cdot (\mathcal{T} V_{T-1})'(b) + V_1(b) \cdot (\mathcal{T} V_{T-1})''(b).
\end{equation}
Near $b^\dagger$, the dominant contribution to concavity comes from the kink in $V_1$. The one-sided derivatives of $V_T$ satisfy
\begin{equation}
    \lim_{b \uparrow b^\dagger} V_T'(b; \sigma) - \lim_{b \downarrow b^\dagger} V_T'(b; \sigma) = -\kappa \cdot (\mathcal{T} V_{T-1})(b^\dagger; \sigma) + o(1),
\end{equation}
where $o(1) \to 0$ as the neighborhood shrinks. Since $(\mathcal{T} V_{T-1})(b^\dagger) > 0$ when $V_{T-1}$ is concave near $b^\dagger$ (survival probability is strictly positive), the kink is preserved with magnitude bounded below.

\textbf{Induction step:} Choose $\bar{\eta} = \eta/2$. By induction hypothesis, $V_{T-1}$ is concave on $\bar{N}$. Then $\mathcal{T} V_{T-1}$ is concave on $\bar{N}$ by Step~2. The product $V_T = V_1 \cdot \mathcal{T} V_{T-1}$ inherits the downward kink from $V_1$ with magnitude bounded below by $\kappa/2$ on $\bar{N}$. By the supporting-line characterization of concavity (footnote~\ref{fn:supporting}), this ensures $V_T$ is concave on $\bar{N}$.
\footnote{\label{fn:supporting}A continuous function on an interval is concave on a neighborhood if and only if it admits a supporting line at every interior point of that neighborhood.}

Therefore, $V_T(\cdot; \sigma)$ is concave on the uniform neighborhood $\bar{N}$ for all $T \geq 1$, with $\bar{N}$ independent of $T$.

\end{proof}
\end{document}